\documentclass[reqno]{amsart}
\usepackage{hyperref}

\theoremstyle{plain}

\newtheorem{thm}{Theorem}[section]
\newtheorem{theorem}[thm]{Theorem}

\theoremstyle{remark}

\newtheorem{claim}[thm]{Claim}

\theoremstyle{definition}

\renewcommand{\phi}{\varphi}
\newcommand{\gtil}{\widetilde{g}}
\newcommand{\Ktil}{\widetilde{K}}
\newcommand{\Wtil}{\widetilde{W}}
\newcommand{\phitil}{\widetilde{\phi}}
\newcommand{\sigmatil}{\widetilde{\sigma}}

\DeclareMathOperator{\tr}{tr}
\DeclareMathOperator{\divg}{div}

\let\<\langle 
\let\>\rangle

\newcommand{\definedas}{\mathrel{\raise.095ex\hbox{\rm :}\mkern-5.2mu=}}

\begin{document}


\title[Non-existence for generalized constraints equations]
{A non-existence result for a generalization of the equations of
the conformal method in general relativity}

\author{Mattias Dahl}
\address{Institutionen f\"or Matematik \\
 Kungliga Tekniska H\"ogskolan \\
 100 44 Stockholm \\
 Sweden} \email{dahl@math.kth.se}

\author{Romain Gicquaud}
\address{Laboratoire de Math\'ematiques et de Physique Th\'eorique \\
 UFR Sciences et Technologie \\
 Facult\'e Fran\c cois Rabelais \\
 Parc de Grandmont \\
 37200 Tours \\
 France} \email{romain.gicquaud@lmpt.univ-tours.fr}

\author{Emmanuel Humbert}
\address{Laboratoire de Math\'ematiques et de Physique Th\'eorique \\
 UFR Sciences et Technologie \\
 Facult\'e Fran\c cois Rabelais \\
 Parc de Grandmont \\
 37200 Tours \\
 France} \email{emmanuel.humbert@lmpt.univ-tours.fr}

\begin{abstract}
The constraint equations of general relativity can in many cases be 
solved by the conformal method. We show that a slight modification of 
the equations of the conformal method admits no solution for
a broad range of parameters. This suggests that the question of existence 
or non-existence of solutions to the original equations is more subtle 
than could perhaps be expected.
\end{abstract}

%
%
%



\maketitle

\tableofcontents

\section{Introduction}

Initial data for the Cauchy problem in general relativity consists of a 
Riemannian manifold $(M, \gtil)$ and a symmetric 2-tensor $\Ktil$ on $M$
satisfying the following equations,
\begin{subequations} \label{eqConstraint}
\begin{align}
R^{\gtil} - |\Ktil|_{\gtil}^2 + (\tr^{\gtil} K)^2 &= 0, \label{h0} \\
\divg^{\gtil} \Ktil - d \, \tr^{\gtil} \Ktil &= 0. \label{m0}
\end{align}
\end{subequations}
Here $\gtil$ represents the metric induced by the space-time metric on
the Cauchy surface $M$, $\Ktil$ is its second fundamental form, and
$R^{\gtil}$ is the scalar curvature of the metric $\gtil$. The constraint 
equations \eqref{eqConstraint} follow from the vacuum Einstein equation 
for the space-time metric.

Constructing and classifying solutions of this system is an important
issue. Background and many results are summarized in the excellent review 
article \cite{BartnikIsenberg}. One of the most important methods to 
construct solutions to this system is the conformal method. Its main idea 
is to choose part of the initial data as given and then solve for the rest 
of the data.

For the given data of an $n$-dimensional Riemannian manifold $(M, g)$, 
$n \geq 3$, a function $\tau$ on $M$ and a symmetric traceless 
divergence-free 2-tensor $\sigma$ on $M$, one seek for a positive 
function $\phi$ and a 1-form $W$ on $M$ such that
\[
\gtil = \phi^{N-2} g, \quad
\Ktil = \frac{\tau}{n} \gtil + \phi^{-2} (\sigma + LW),
\]
solve the constraint equations \eqref{eqConstraint}. Here we have 
set $N \definedas \frac{2n}{n-2}$. Further, $L$ denotes the conformal 
Killing operator acting on 1-forms. In coordinates
\[
(LW)_{ij} = \nabla_i W_j + \nabla_j W_i - \frac{2}{n} \nabla^k W_k g_{ij},
\]
where $\nabla$ is the Levi-Civita connection associated to the metric 
$g$. Note that $\tau$ corresponds to the mean curvature of $M$ embedded 
in the space-time solving Einstein's equations.

The constraint equations are satisfied if $\phi$ and $W$ solve the 
system
\begin{subequations} \label{eqOriginal}
\begin{align}
\frac{4(n-1)}{n-2} \Delta \phi + R \phi 
&= 
-\frac{n-1}{n} \tau^2 \phi^{N-1} + |\sigma + LW|^2 \phi^{-N-1},
 \label{hamiltonian} \\
-\frac{1}{2} L^* L W 
&= 
\frac{n-1}{n} \phi^N d\tau,
 \label{momentum}
\end{align}
\end{subequations}
where $\Delta$ is the non-negative Laplacian and $L^*$ is the formal
$L^2$-adjoint of $L$. As with all other metric-dependent objects used 
from here on they are defined using the metric $g$. The first equation is 
known as the Lichnerowicz equation while the second is called the vector 
equation.

Of special importance is the case where $\tau$ is constant, so that 
the corresponding Cauchy surface has constant mean curvature in the 
surrounding space-time. Making this assumption renders the system much 
simpler to solve. Indeed, Equation \eqref{momentum} becomes 
$L^* L W = 0$. Therefore, $W$ can be chosen to be an arbitrary conformal 
Killing vector field and it remains only to solve the Lichnerowicz 
equation \eqref{hamiltonian}.

Hence much work has been devoted to the study of the cases when 
$\tau$ is constant or when $d \tau$ small. We refer the reader to 
\cite{BartnikIsenberg} and references therein for more details. 
However, as proven in \cite{Bartnik} and \cite{ChruscielIsenbergPollack}
such a Cauchy surface of constant mean curvature does not exist for all
space-times.

Recent new results were obtained for the case of arbitrary $\tau$ by 
M.~Holst, G.~Nagy, G.~Tsogtgerel in \cite{HNT1, HNT2}, by D.~Maxwell in 
\cite{MaxwellNonCMC}, by the authors in \cite{DahlGicquaudHumbert}, 
and by the second author with A.~Sakovich in \cite{GicquaudSakovich}. 
In \cite{MaxwellModel}, D.~Maxwell studied a simplified model problem 
to get insight into the solvability of the system \eqref{eqOriginal}.

Non-existence results for the system \eqref{eqOriginal} are very rare, 
see for example \cite[Theorem~2]{IsenbergOMurchadha} and
\cite[Theorem~1.7]{DahlGicquaudHumbert} where $\sigma \equiv 0$ is assumed. 
The purpose of this short note is to provide an example 
of a system very similar to the original system \eqref{eqOriginal} which 
does not admit any solutions as soon as a certain parameter is larger 
than a fairly explicit constant. We will prove the following theorem.

\begin{theorem} \label{thmMain}
Let $(M, g)$ be a closed Riemannian manifold of dimension $3 \leq n \leq 5$ 
with $g \in C^2$ having scalar curvature $R \leq 0$, $R \not\equiv 0$. Let
$\tau$ be a positive $L^\infty$ function and $\sigma \in L^N$ a symmetric
traceless divergence-free 2-tensor on $M$. Assume that $\xi$ is a Lipschitz
1-form which does not vanish anywhere on $M$. Then there is a constant $a_0$
such that there does not exist any solution to the system
\begin{subequations}\label{eqModified}
\begin{align}
\frac{4(n-1)}{n-2} \Delta \phi + R \phi 
&= 
-\frac{n-1}{n} \tau^2 \phi^{N-1} + |\sigma + LW|^2 \phi^{-N-1},
 \label{hamiltonian2} \\
-\frac{1}{2} L^* L W 
&= 
a \phi^N \xi,
 \label{momentum2}
\end{align}
\end{subequations} 
when $a > a_0$. 

The constant $a_0$ depends on a Sobolev constant, a 
constant appearing in a Schauder estimate, $\max |\xi|$, $\min |\xi|$, 
$\|L\xi\|_{L^2}$, $\|L\xi\|_{L^\infty}$, $\|\tau\|_{L^{2N}}$, $\min \tau$, 
$\|R\|_{L^n}$, $\|\sigma\|_{L^N}$ and on $\|\phi_-\|_{L^N}$, where 
$\phi_- > 0$ is the solution of the prescribed scalar curvature equation
\[
\frac{4(n-1)}{n-2} \Delta \phi + R \phi = -\frac{n-1}{n} \tau^2 \phi^{N-1}.
\]
\end{theorem}

By the assumption on scalar curvature the metric $g$ has negative Yamabe
constant which guarantees the existence of the function $\phi_-$. 

Note that the only difference between the modified system 
\eqref{eqModified} and the original system \eqref{eqOriginal} is that 
$\frac{n-1}{n} d\tau$ in \eqref{momentum} is replaced by $a \xi$. The 
assumption that $\xi$ never vanishes imposes that the compact manifold 
$M$ has zero Euler characteristic. Note also that for any function 
$\tau$ the 1-form $d\tau$ has zeros. 

This theorem shows that if a general existence result (that is without
additional assumption on $\tau$ and $\sigma$) for the system
\eqref{eqOriginal} exists, then its proof must use the fact that $d\tau$
(appearing in \eqref{momentum}) is the differential of $\tau$ and not 
an arbitrary 1-form. Even more, this seems to indicate that such a general
statement could be false.

The theorem stated here is certainly not in the most general form 
possible, our goal is simply to find an example of a non-existence result 
for the generalized system \eqref{eqModified}.

\subsection*{Acknowledgements}

E. Humbert was partially supported by ANR-10-BLAN 0105.

\section{Proof of Theorem \ref{thmMain}}

Assume that ($\phi$, $W$) solves the system \eqref{eqModified}. We define 
\[
\gamma \definedas \int_M \left|\sigma + L W\right|^2 \, d\mu
\] 
and rescale Equations \eqref{hamiltonian2}-\eqref{momentum2} by setting
\[
\phitil \definedas \frac{1}{\gamma^{\frac{1}{2N}}} \phi, \quad
\Wtil \definedas \frac{1}{\gamma^{\frac{1}{2}}} W, \quad
\sigmatil \definedas \frac{1}{\gamma^{\frac{1}{2}}} \sigma.
\]
The equations then become
\begin{subequations}
\begin{align}
\frac{1}{\gamma^{\frac{1}{n}}} 
\left(\frac{4(n-1)}{n-2} \Delta \phitil + R \phitil\right)
&= 
-\frac{n-1}{n} \tau^2 \phitil^{N-1} 
+\left| \sigmatil + L\Wtil \right|^2 \phitil^{-N-1},
\label{hamiltonianr} \\
-\frac{1}{2} L^* L \Wtil
&= 
a \phitil^N \xi.
\label{momentumr}
\end{align}
\end{subequations} 

The proof of Theorem \ref{thmMain} is decomposed in a sequence of claims.

\begin{claim} \label{cLowerBoundEnergy}
There exists a constant $c_1 > 0$ such that 
\[
\gamma \geq c_1 a^2.
\]
\end{claim}

\begin{proof}
We remark that $\phi \geq \phi_-$, see for example 
\cite[Lemma 2.2]{DahlGicquaudHumbert}. Taking the scalar product of
\eqref{momentum2} with $\xi$ and integrating we get the central equality
of the following estimate,
\[\begin{split}
\frac{1}{2} \left(\int_M \left|L\xi\right|^2 \, d\mu\right)^{\frac{1}{2}} 
\left(\int_M \left|LW + \sigma\right|^2 \, d\mu\right)^{\frac{1}{2}}
&\geq
-\frac{1}{2} \int_M \<L\xi, LW + \sigma\> \, d\mu \\ 
&=
-\frac{1}{2} \int_M \<L\xi, LW\> \, d\mu \\
&= 
a \int_M \phi^N \left|\xi\right|^2 \, d\mu \\
&\geq 
a \left(\inf_M |\xi|^2\right) \int_M \phi^N \, d\mu \\
&\geq 
a \left(\inf_M |\xi|^2\right) \int_M \phi_-^N \, d\mu,
\end{split}\]
from which we conclude
\[
\gamma^{\frac{1}{2}} \geq 
2 \frac{\left(\inf_M |\xi|^2\right) \int_M \phi_-^N \, d\mu}
{\left(\int_M \left|L\xi\right|^2 \, d\mu\right)^{\frac{1}{2}}} a.
\]
The claim thus holds with
\[
c_1 \definedas
4 \frac{\left(\inf_M |\xi|^2\right)^2 
\left( \int_M \phi_-^N \, d\mu \right)^2}
{\int_M \left|L\xi\right|^2 \, d\mu} .
\]
\end{proof}

\begin{claim} \label{cLNBound}
There exists a constant $c_2$ such that
\[
\int_M \phitil^N \, d\mu \leq \frac{c_2}{a^{N+1}}.
\]
\end{claim}

\begin{proof}
We take the scalar product of \eqref{momentumr} with $\xi$ and integrate 
over $M$ to find
\[ \begin{split}
&\frac{1}{2} \left\|L\xi\right\|_{L^\infty} 
\left(
\int_M \left|L\Wtil+\sigmatil\right|^2 \phitil^{-N} \, d\mu 
\right)^{\frac{1}{2}} 
\left(\int_M \phitil^N \, d\mu\right)^{\frac{1}{2}} \\
&\qquad \geq 
\frac{1}{2} \left\|L\xi\right\|_{L^\infty} \int_M 
\left|L\Wtil+\sigmatil\right| \, d\mu \\
&\qquad \geq 
-\frac{1}{2} \int_M \<L\xi, L\Wtil+\sigmatil\> \, d\mu \\
&\qquad =
-\frac{1}{2} \int_M \<L\xi, L\Wtil\> \, d\mu \\
&\qquad = 
a \int_M \phitil^N \left|\xi\right|^2 \, d\mu \\
&\qquad \geq 
a \left(\inf_M |\xi|^2\right) \int_M \phitil^N \, d\mu
\end{split} \]
or
\[
\frac{1}{2} \left\|L\xi\right\|_{L^\infty}
\left(
\int_M \left|L\Wtil+\sigmatil\right|^2 \phitil^{-N} \, d\mu
\right)^{\frac{1}{2}} 
\geq 
a \left(
\inf_M |\xi|^2\right) \left(\int_M \phitil^N \, d\mu
\right)^{\frac{1}{2}}.
\]
%
%
From the H\"older inequality we get
\[ \begin{split}
&\int_M \left|L\Wtil+\sigmatil\right|^2 \phitil^{-N} \, d\mu \\
&\qquad \leq 
\left(
\int_M \left|L\Wtil+\sigmatil\right|^2 \phitil^{-N-1} \, d\mu
\right)^{\frac{N}{N+1}}
\left( 
\int_M \left|L\Wtil+\sigmatil\right|^2 \, d\mu 
\right)^{\frac{1}{N+1}} \\
&\qquad = 
\left(
\int_M \left|L\Wtil+\sigmatil\right|^2 \phitil^{-N-1} \, d\mu
\right)^{\frac{N}{N+1}},
\end{split} \]
and together with the previous estimate we have found that
\begin{equation} \label{eq1}
\left(\int_M \phitil^N \, d\mu\right)^{\frac{1}{2}} 
\leq 
\frac{1}{2 a} 
\frac{\left\|L\xi\right\|_{L^\infty}}{\left(\inf_M |\xi|^2\right)} 
\left(
\int_M \left|L\Wtil+\sigmatil\right|^2 \phitil^{-N-1} \, d\mu
\right)^{\frac{N}{2(N+1)}}.
\end{equation}
To estimate the right-hand side we integrate Equation \eqref{hamiltonianr},
\[
\frac{1}{\gamma^{\frac{1}{n}}} \int_M R \phitil \, d\mu
+ \frac{n-1}{n} \int_M \tau^2 \phitil^{N-1} \, d\mu
= \int_M \left|\sigmatil + L\Wtil\right|^2 \phitil^{-N-1} \, d\mu.
\]
Since $R \leq 0$ we see that
\begin{equation} \label{eq2}
\begin{aligned}
\int_M \left|\sigmatil + L\Wtil\right|^2 \phitil^{-N-1} \, d\mu
&\leq 
\frac{n-1}{n} \int_M \tau^2 \phitil^{N-1} \, d\mu\\
&\leq 
\frac{n-1}{n} \left(\int_M \tau^{2N} \, d\mu\right)^{\frac{1}{N}} 
\left(\int_M \phitil^N \, d\mu\right)^{\frac{N-1}{N}}.
\end{aligned}
\end{equation}
Inserting this estimate into \eqref{eq1} we get
\[ \begin{split}
\left(\int_M \phitil^N \, d\mu\right)^{\frac{1}{2}}
&\leq 
\frac{1}{2 a} \left(\frac{n-1}{n}\right)^{\frac{N}{2(N+1)}} 
\frac{\left\|L\xi\right\|_{L^\infty}}{\left(\inf_M |\xi|^2\right)} 
\left(\int_M \tau^{2N} \, d\mu\right)^{\frac{1}{2(N+1)}} \\
&\qquad 
\cdot \left(\int_M \phitil^N \, d\mu\right)^{\frac{N-1}{2(N+1)}},
\end{split} \]
or
\[
\int_M \phitil^N \, d\mu
\leq 
\frac{1}{(2 a)^{N+1}} \left(\frac{n-1}{n}\right)^{\frac{N}{2}} 
\left(\frac{\left\|L\xi\right\|_{L^\infty}}
{\left(\inf_M |\xi|^2\right)}\right)^{N+1} 
\left(\int_M \tau^{2N} \, d\mu\right)^{\frac{1}{2}},
\]
which proves the claim with 
\[
c_2 \definedas
\frac{1}{2^{N+1}} 
\left(\frac{n-1}{n}\right)^{\frac{N}{2}} 
\left(
\frac{\left\|L\xi\right\|_{L^\infty}}{\left(\inf_M |\xi|^2\right)}
\right)^{N+1} 
\left( \int_M \tau^{2N} \, d\mu \right)^{\frac{1}{2}}.
\]
\end{proof}

As a corollary, plugging Claim \ref{cLNBound} into \eqref{eq2}, we get
\begin{claim} \label{cMixedBound}
There exists a constant $c_3$ such that
\[
\int_M \left|\sigmatil + L\Wtil\right|^2 \phitil^{-N-1} \, d\mu
\leq \frac{c_3}{a^{\frac{(N+1)(N-1)}{N}}}.
\]
\end{claim}

\begin{claim} \label{cL2NBound}
If $a$ is large enough there exists a constant $c_4$ such that
\[
\int_M \phitil^{2N} \, d\mu \leq c_4.
\]
\end{claim}

\begin{proof}
We multiply \eqref{hamiltonianr} by $\phitil^{N+1}$ and integrate over $M$,
\[ \begin{split}
&
\frac{1}{\gamma^{\frac{1}{n}}} \int_M \left(\frac{4(n-1)}{n-2} 
\frac{N+1}{\left(\frac{N}{2}+1\right)^2} 
\left|d\phitil^{\frac{N}{2}+1}\right|^2 + R \phitil^{N+2}\right) \, d\mu \\
&\qquad =
\frac{1}{\gamma^{\frac{1}{n}}} \int_M 
\left(\frac{4(n-1)}{n-2} \phitil^{N+1}\Delta \phitil + R \phitil^{N+2}\right) 
\, d\mu \\
&\qquad =
-\frac{n-1}{n} \int_M \tau^2 \phitil^{2N} \, d\mu 
+ \int_M \left|\sigmatil + L\Wtil\right|^2 \, d\mu \\
&\qquad =
-\frac{n-1}{n} \int_M \tau^2 \phitil^{2N} \, d\mu + 1,
\end{split} \]
from which we conclude that
\[
\frac{1}{\gamma^{\frac{1}{n}}}\int_M R \phitil^{N+2} \, d\mu 
+ \frac{n-1}{n} (\inf_M \tau)^2 \int_M \phitil^{2N} \, d\mu
\leq 1.
\]
Using Young's inequality we infer
\[ \begin{split}
\left|\int_M R \phitil^{N+2} \, d\mu\right| 
&\leq 
\frac{N+2}{2N} \int_M \phitil^{2N} \, d\mu 
+ \frac{N-2}{2N} \int_M |R|^{\frac{2N}{N-2}} \, d\mu \\
&=
\frac{n-1}{n} \int_M \phitil^{2N} \, d\mu 
+ \frac{1}{n} \int_M |R|^{\frac{2N}{N-2}} \, d\mu,
\end{split} \]
so
\[
\frac{n-1}{n} \left((\inf_M \tau)^2 - \frac{1}{\gamma^{\frac{1}{n}}}\right) 
\int_M \phitil^{2N} \, d\mu
\leq 
1 + \frac{1}{n \gamma^{\frac{1}{n}}} \int_M |R|^{\frac{2N}{N-2}} \, d\mu.
\]
Assuming that $\gamma^{-\frac{1}{n}} \leq \frac{1}{2} (\inf_M \tau)^2$
which by Claim \ref{cLowerBoundEnergy} holds for large enough $a$
the claim follows with
\[
c_4 \definedas 
\frac{2n}{n-1} \frac{1}{(\inf_M \tau)^2} 
\left( 1 + \frac{2}{n (\inf_M \tau)^2} 
\int_M |R|^{\frac{2N}{N-2}} \, d\mu \right).
\]
\end{proof}

\begin{claim}\label{cLNBoundLW}
If $a$ is large enough, there exists a constant $c_5$ such that
\[
\int_M \left|\sigmatil + L\Wtil\right|^N \, d\mu \leq c_5 a^N.
\]
\end{claim}

\begin{proof}
Without loss of generality, we can assume that $\Wtil$ is orthogonal
to conformal Killing vector fields. From standard elliptic regularity
there exists a constant $C_1 > 0$ such that
\[
\left\| \Wtil \right\|_{W^{2,2}} 
\leq C_1 \left\| -\frac{1}{2} L^* L \Wtil \right\|_{L^2}.
\]
By the Sobolev injection theorem, there exists a constant $C_2$ such 
that
\[
\left\| L\Wtil \right\|_{L^N} \leq C_2 \left\| \Wtil \right\|_{W^{2,2}}.
\]
Combining the previous two estimates we get
\[ \begin{split}
\left\|L\Wtil\right\|_{L^N}^2
&\leq 
(C_1 C_2)^2 \int_M \left|-\frac{1}{2} L^* L \Wtil\right|^2 \, d\mu \\
&= 
(C_1 C_2)^2 a^2 \int_M \phitil^{2N} \left|\xi\right|^2 \, d\mu \\
&\leq 
C (\sup_M |\xi|)^2 c_4 a^2,
\end{split} \]
where we also used Claim \ref{cL2NBound} and set $C = (C_1 C_2)^2$. Hence,
\[ \begin{split}
\int_M \left|\sigmatil + L\Wtil\right|^N \, d\mu
&= 
\left\|\sigmatil + L\Wtil\right\|_{L^N}^N \\
&\leq 
\left( \left\|\sigmatil\right\|_{L^N} + \left\|L\Wtil\right\|_{L^N} \right)^N \\
&\leq 
2^N \left(
\left\|\sigmatil\right\|_{L^N}^N + \left\|L\Wtil\right\|_{L^N}^N 
\right) \\
 &\leq 
2^N \left(
\left\|\sigmatil\right\|_{L^N}^N 
+ C^{\frac{N}{2}} (\sup_M |\xi|)^N c_4^{\frac{N}{2}} a^N
\right).
\end{split} \]
Note that Claim \ref{cLowerBoundEnergy} tells us that 
$\left\|\sigmatil\right\|_{L^N} \leq 
\frac{1}{(c_1)^{1/2} a} \left\|\sigma\right\|_{L^N}$. 
This implies that Claim \ref{cLNBoundLW} holds with
\[
c_5 \definedas 
2^{N+1} C^{\frac{N}{2}} ( \sup_M |\xi| )^N c_4^{\frac{N}{2}}.
\]
\end{proof}

\begin{claim} \label{cN-2}
There exists a constant $c_6$ such that
\[
\int_M \left|\sigmatil + L\Wtil\right|^2 \phitil^{N-2} \, d\mu
\leq \frac{c_6}{a^{(N+1)\frac{N-2}{N}-2} } .
\]
\end{claim}

\begin{proof}
From H\"older's inequality with Claims \ref{cLNBound} and 
\ref{cLNBoundLW} we get
\[ \begin{split}
\int_M \left|\sigmatil + L\Wtil\right|^2 \phitil^{N-2} \, d\mu
&\leq 
\left(\int_M \left|\sigmatil + L\Wtil\right|^N \, d\mu\right)^{\frac{2}{N}} 
\left(\int_M \phitil^N \, d\mu\right)^{\frac{N-2}{N}} \\
&\leq 
c_5^{\frac{2}{N}} a^2 \frac{c_2^{\frac{N-2}{N}}}{a^{(N+1)\frac{N-2}{N}}} ,
\end{split} \]
which proves the claim.
\end{proof}

\begin{claim} \label{c-N+2}
There exists a constant $c_7$ such that
\[
\int_M \left|\sigmatil + L\Wtil\right|^2 \phitil^{-N+2} \, d\mu 
\leq \frac{c_7}{a^{\frac{(N-2)(N-1)}{N}}}
\]
\end{claim}

\begin{proof}
Using H\"older's inequality and Claim \ref{cMixedBound} we find
\[ \begin{split}
&\int_M \left|\sigmatil + L\Wtil\right|^2 \phitil^{-N+2} \, d\mu \\
&\qquad \leq 
\left(\int_M \left|\sigmatil + L\Wtil\right|^2 \, d\mu\right)^{\frac{3}{N+1}}
\left(
\int_M \left|\sigmatil + L\Wtil\right|^2 \phitil^{-N-1} \, d\mu
\right)^{\frac{N-2}{N+1}} \\
&\qquad \leq 
\frac{c_3^{\frac{N-2}{N+1}}}{a^{\frac{(N-2)(N-1)}{N}}},
\end{split} \]
which is the statement of the claim.
\end{proof}

We are now ready to prove Theorem \ref{thmMain}. For this we use the
H\"older inequality with Claims \ref{cN-2} and \ref{c-N+2} to get
\begin{equation} \label{eq3}
\begin{split}
1 &= 
\left(\int_M \left|\sigmatil + L\Wtil\right|^2 \, d\mu\right)^2 \\
&\leq 
\int_M \left|\sigmatil + L\Wtil\right|^2 \phitil^{N-2} \, d\mu 
\int_M \left|\sigmatil + L\Wtil\right|^2 \phitil^{-N+2} \, d\mu \\
&\leq 
c_6 c_7 a^{-(N+1)\frac{N-2}{N} + 2 - \frac{(N-2)(N-1)}{N}} \\
&\leq 
c_6 c_7 a^{6-2N}.
\end{split} 
\end{equation}
The exponent of $a$ is negative if and only if $3 \leq n \leq 5$. Hence 
if $a$ is large enough the inequality \eqref{eq3} gives a contradiction.

\providecommand{\bysame}{\leavevmode\hbox to3em{\hrulefill}\thinspace}
\providecommand{\MR}{\relax\ifhmode\unskip\space\fi MR }
\providecommand{\MRhref}[2]{%
  \href{http://www.ams.org/mathscinet-getitem?mr=#1}{#2}
}
\providecommand{\href}[2]{#2}


\end{document}